\documentclass[10pt,a4paper]{article}
\usepackage[margin=23mm]{geometry}
\usepackage[T1]{fontenc}
\usepackage{lmodern,amsmath,amssymb,amsthm,mathtools,booktabs,microtype,needspace}
\usepackage{xpatch}
\xpatchcmd{\proof}{\itshape}{\normalfont}{}{\PackageError{mn-dna}{Cannot set upright proof headings}{Check the amsthm proof definition.}}
\usepackage[hidelinks]{hyperref}
\usepackage[nameinlink,capitalise,noabbrev]{cleveref}
\theoremstyle{definition}
\newtheorem{theorem}{Theorem}
\newtheorem{lemma}{Lemma}
\newtheorem{proposition}{Proposition}

\theoremstyle{definition}
\newtheorem{definition}{Definition}
\theoremstyle{definition}

\newcommand{\F}{\mathbb F_2}
\newcommand{\E}{\mathbb E}
\newcommand{\Prob}{\mathbb P}
\newcommand{\C}{\mathcal C}
\newcommand{\T}{\mathcal T}
\newcommand{\ones}{\boldsymbol 1}

\DeclareMathOperator{\wt}{wt}
\DeclareMathOperator{\Var}{Var}
\title{MacKay--Neal Codes for DNA Channels\\under Block-MAP Decoding}
\author{Kenta Kasai\\\small Institute of Science Tokyo}
\date{October 28, 2026}
\hypersetup{pdftitle={MacKay\textendash{}Neal Codes for DNA Channels under Block-MAP Decoding},pdfauthor={Kenta Kasai}}
\begin{document}
\maketitle
\begin{abstract}
We prove achievability below the symmetric information rate (SIR) under block maximum a posteriori (MAP) decoding for a fixed-degree MacKay--Neal (MN) socket ensemble on the quaternary insertion--deletion--substitution channel used to model DNA reads. For every integer $\ell\ge4$, the $(\ell,3,3)$ socket ensemble, with one transmitted bit fixed and an independent uniform known offset, has limiting rate $r_\star=6/\ell$ bits per stored nucleotide. If the joint SIR $J_t$ of a fixed number $t$ of reads under independent, uniformly distributed input nucleotides exceeds $r_\star$, the ensemble-average full extended-word block error is $O(1/n)$, and deterministic graph--offset sequences attain the same limiting rate and message-average error bound. This result concerns the discrete rate family $6/\ell$, not arbitrary rates approaching a given SIR. The channel permits arbitrarily many insertions and unbounded drift, with no alignment supplied to the receiver. The proof combines established MN parity and distance estimates with a strict arbitrary-candidate-list bound and a list construction that accounts for hidden output boundaries.
\end{abstract}

\section{Introduction}
Insertions and deletions obscure the correspondence between a stored DNA sequence and its reads. A standard idealized insertion--deletion--substitution (IDS) model processes each input nucleotide by repeated random insertions, followed by either deletion or a possibly substituted transmission. The receiver observes the resulting sequence without knowing which output symbols correspond to each input nucleotide. Several conditionally independent reads of the same stored sequence can be decoded jointly~\cite[Sec.~II-A and Sec.~IV]{Maarouf2023}. Its event probabilities are independent of the input symbols and surrounding sequence, and inserted symbols are uniform; these assumptions specify the idealization used here.

Practical constructions combine synchronization recovery with error correction. Watermark codes provide soft synchronization information to an outer nonbinary low-density parity-check (LDPC) decoder~\cite{DaveyMacKay2001}; sparse-graph synchronization decoders have also been studied by density evolution~\cite{Goto2018}. For multiple quaternary reads, inner convolutional or time-varying block codes can be concatenated with outer LDPC or polar codes~\cite[Secs.~II--VII]{Maarouf2023}. A recent half-marker construction uses one bit of selected quaternary symbols for synchronization while retaining the other for data. Its information-rate and finite-length experiments assess this allocation and iterative decoding with an outer LDPC code~\cite{HaghighatDuman2025}.

Recent work also places error-correcting constraints across many short strands. Cross-strand polar coding with decision-feedback shifts is evaluated for deletion-only and insertion-only channels in~\cite[Secs.~IV--V]{Lin2025}. A subsequent trellis construction handles simultaneous IDS errors by alternating posterior estimation along binary strands with polar decoding across strands~\cite[Secs.~IV--V]{Wu2026}; its near-capacity evidence is numerical and uses a conjectured capacity benchmark. These cross-strand studies assume one read from each strand, with the correspondence between reads and strands known; the present theorem instead lets the length of an identified strand grow and jointly uses a fixed number of its reads.

Learned polar decoders provide another direction: channel-sample training is combined with successive-cancellation or list decoding and evaluated on synthetic IDS channels, multiple reads, and nanopore data~\cite{AharoniPfister2025}. At the storage-system level, a recent end-to-end scheme combines inner, index, and outer codes with clustering and soft-information processing for sampled nanopore reads~\cite{Welter2026}. Its symbolwise posterior calculations are components of a practical receiver, whereas our block maximum a posteriori (MAP) decoding criterion concerns recovery of the entire extended codeword. Pool sampling, strand identification, and sequence-dependent noise broaden the engineering model in that work; they are not assumptions of our identified-strand theorem.

Provable information-rate-achieving constructions also precede this work. On the binary deletion channel, polar codes with known separating sequences, called guard bands, achieve the information rate of a regular hidden-Markov input process; optimizing such processes yields capacity achievability~\cite{Tal2022}. A subsequent construction strengthens the error-decay guarantee with polynomial-time decoding~\cite{AravaTal2025}. For a related binary IDS model, modified guard bands likewise permit polar codes to achieve the specified input process's information rate with polynomial encoding and decoding complexity~\cite[Thm.~1 and Sec.~III]{PfisterTal2021}. Our result concerns a different code structure and decoding guarantee: a prescribed fixed-degree MacKay--Neal (MN) ensemble on quaternary unaligned reads under matched block-MAP decoding.

Random-coding theorems for synchronization channels have been known since~\cite{Dobrushin1967}. For independent and identically distributed (iid) uniform input nucleotides, the symmetric information rate (SIR) is the limiting mutual information per stored nucleotide between the input and all observed reads. In our model, independently and uniformly chosen codewords already yield vanishing average block error at every fixed rate strictly below the joint SIR. Indeed, the candidate lists of \cref{lem:channel-list} contain the true input with probability tending to one, and an independent competitor belongs to a list with probability equal to its size divided by $4^n$. A union bound then proves random-code achievability. The additional question is whether this guarantee survives common sparse parity constraints, under which different codewords are no longer independently chosen.

We prove that the $(\ell,3,3)$ MN socket ensemble, with one bit fixed and a known uniform offset, has vanishing full extended-word block-MAP error whenever its limiting rate $6/\ell$ lies below the joint iid-uniform SIR. The degrees remain fixed as the blocklength grows. The contribution is this reliability theorem for the specified ensemble and discrete rate family; it does not improve the achievable rates of unconstrained random coding or establish an efficient MAP implementation.

The SIR is determined by the channel law, the number of reads, and the iid-uniform input law. Evaluating it is a separate task from proving reliability of this MN ensemble. For general IDS parameters, the information-rate analysis in~\cite[Sec.~VI-B]{Maarouf2023} uses numerical evaluation through channel-likelihood recursions. Our achievability theorem requires no closed-form expression for the SIR: a rigorous lower bound exceeding the ensemble's limiting rate suffices to verify its condition. \Cref{app:noise} gives a concrete application using an analytic lower bound.

The proof combines established techniques with two additional estimates. Weight-distribution methods transfer random-coding bounds to structured code ensembles~\cite{ShulmanFeder1999}, and random cosets are used in LDPC achievability on arbitrary discrete memoryless channels~\cite{BennatanBurshtein2004}. A random coset adds one known random offset to every codeword. Revealing output-block boundaries and charging their entropy has precedents in deletion-channel information-rate analysis~\cite{FertonaniDuman2010}\cite[Lemma~1]{Dalai2011}. Here a strict MN candidate-counting estimate is combined with a posterior-threshold list that considers every partition of the unaligned reads. The inherited MN results and the additional steps are specified below; their supporting sources are also cited where each technique is used.

The analysis uses the observations available to the receiver. No insertion-count or drift bound is imposed on the channel, where drift is the difference between output and input positions. Although a comparison list discards exceptionally long outputs, the matched MAP likelihood still sums all paths consistent with each observed read. Thus the sufficient reliability condition uses the SIR of the untruncated channel.

\subsection{Main theorem: MAP achievability below the SIR}
The code has $N=2n$ transmitted binary variables, grouped into $n$ DNA symbols, $K=3N/\ell$ punctured variables that are not transmitted, and $N$ parity checks. A socket is an edge endpoint before matching. Each check has three sockets for each variable type. Punctured variables have degree $\ell$ and transmitted variables degree three; the degree counts their incident edges. The two socket matchings are uniform and independent, and define the binary matrices $A$ and $B$ connecting the two variable types to the checks. The extended code, which includes both transmitted and punctured variables, is the kernel of $[A\ B]$, and its transmitted projection is $\C(G)$. We impose $v_1=0$, retain that coordinate in transmission, and add a binary vector known to the encoder and decoder, called the offset. This single constraint costs exactly one information bit because the unrestricted code contains the all-one word. No invertibility condition is imposed on $B$.

\Needspace{25\baselineskip}
\begin{theorem}[Block-MAP achievability below the SIR]\label{thm:main}
Fix an integer $\ell\ge4$, a positive integer $t$, and channel parameters satisfying \cref{def:channel}. The $t$ reads are independent conditional on their common input. Let $J_t$ be the joint iid-uniform SIR in bits per stored nucleotide, and set $R=3/\ell$ and $r_\star=2R=6/\ell$.

Consider the MN ensemble of \cref{def:ensemble}, with one fixed transmitted bit retained in transmission and an independent uniform offset known to the decoder. As $n\to\infty$ through $K=6n/\ell\in\mathbb Z$, its actual transmitted rate tends to $r_\star$ in probability and in $L^1$. 

The error is averaged over graphs, offsets, uniform extended codewords, and channel outputs.
If $r_\star<J_t$, the average block-MAP error for the full extended word is $O(1/n)$. Deterministic graph--offset sequences attain rate $r_\star$ with message-average block error $O(1/n)$.

The degrees, channel parameters, and number of reads are fixed before taking the limit. The case $r_\star=J_t$ is left open.
\end{theorem}

The theorem applies to the discrete rate family $6/\ell$, $\ell\ge4$, with a fixed positive gap $J_t-6/\ell$; it does not assert arbitrary rates below a prescribed SIR. The $O(1/n)$ constant may depend on this gap and can be large, so the bound does not predict the blocklength needed for a target error rate. A gap shrinking with $n$ is a different limiting regime. The SIR specifies an iid-uniform input law, whereas unrestricted channel capacity optimizes the input process. The theorem establishes reliability under MAP decoding; computationally efficient algorithms for decoding and for selecting deterministic codes require further results. Correct block recovery includes every punctured and transmitted bit, so the corresponding bit-error probabilities also tend to zero.

The limiting rate is measured per stored nucleotide, without division by the number of reads. The model concerns reads of an identified strand. Sampling from a pool, unknown strand identities, and strand loss require the additional indexing and sampling models of~\cite{Welter2026,Shomorony2021}. Sequence-dependent nanopore noise~\cite{Welter2026} and biochemical sequence constraints also require additional analysis.

\subsection{Extension from memoryless channels to synchronization errors}
MN codes use punctured variables to obtain dense effective constraints from a sparse graph~\cite[Sec.~VI-B]{MacKay1999}. The closest prior result is the fixed-degree capacity theorem for binary-input memoryless symmetric channels in~\cite[Thm.~1]{KasaiMAP2026}. That theorem already treats the same $(\ell,3,3)$ socket ensemble, block-MAP recovery of both punctured and transmitted variables, singular $B$, and convergence of the actual rate. These properties are inherited here. We reuse its parity entropy inequality, distance estimate, and exact socket-counting formula~\cite[Lemmas~1 and~8, Sec.~4.2]{KasaiMAP2026}; \cref{prop:rate} gives a direct rate proof from the same count. The additional result concerns synchronization errors: the same underlying graph ensemble, with one bit fixed and a known offset added, has vanishing block error when its limiting rate lies below the joint iid-uniform SIR of the unaligned DNA reads.

The memoryless-channel theorem cannot simply be applied after replacing capacity by SIR. Its comparison with an equal-capacity binary symmetric channel replaces observations one input coordinate at a time and relies on the product structure of a memoryless channel~\cite[Sec.~4.3, Lemmas~5--6]{KasaiMAP2026}. Although the IDS model processes input symbols independently, the emitted strings are concatenated and their boundaries are hidden. The receiver therefore lacks the input-indexed observations required by that comparison. A uniform-input information rate alone does not supply the channel comparison needed for correlated codeword coordinates.

The ensemble part of the extension is a candidate-counting estimate that does not assume a binary symmetric channel model. The proof for the binary symmetric channel in~\cite[Sec.~4.2]{KasaiMAP2026} counts noise vectors of a fixed Hamming weight and bounds their conditional entropy. Here \cref{lem:gap,lem:list} apply to an arbitrary candidate input list when the input and list are jointly independent of the graph. If its size is exponentially smaller than $2^{N(1-R)}$, the expected number of wrong candidates at intermediate Hamming distances is exponentially small. The proof requires the exponent to be uniformly strictly negative, including at the endpoints for the punctured-variable weight, rather than merely nonpositive. This controls the probability of a competing codeword directly.

\subsection{Proof outline}
The proof follows the receiver's task. First, \cref{sec:lists} constructs the set $\T_n(y^{[t]})$ of candidate input sequences whose full blocks satisfy the posterior-product threshold under at least one partition of the observed reads. Its size is at most $2^{2n\eta}$, and the true sequence is absent with probability $O(1/n)$, for a fixed $\eta>1-J_t/2$. The cost of trying all partitions becomes arbitrarily small per input symbol when the fixed block size is increased.

Next, \cref{sec:counting} asks whether another MN codeword can remain in that set. Fixing one bit removes the all-one difference, the distance estimate excludes differences near both extreme weights, and the socket count makes the expected number of the remaining competitors exponentially small when $\eta<1-R$. Averaging the uniform offset is what permits applying this graph count to a list chosen from channel observations.

The comparison decoder keeps only candidates satisfying the code constraints. On a good graph it succeeds whenever the true word is present and no competitor remains. MAP has no larger block error. Finally, \cref{sec:map-proof} combines the two conditions on $\eta$, which are compatible because $J_t>2R$, and selects deterministic graph--offset pairs with both vanishing error and the required actual rate. The supporting rate calculation is in \cref{app:rate}. \Cref{sec:model} supplies the definitions used in these steps.

\section{Channel, code, and rates}\label{sec:model}
\begin{definition}[Quaternary queue IDS channel]\label{def:channel}
The alphabet is $\mathcal A=\{\mathrm A,\mathrm C,\mathrm G,\mathrm T\}$. Fix a bijection $\beta:\mathcal A\to\F^2$ and write $\beta_n:\mathcal A^n\to\F^{2n}$ for its coordinatewise extension. Fix $p_I,p_D\ge0$ with $p_I+p_D\le1$ and $p_I<1$, and $p_S\in[0,1]$. While an input symbol is queued, an insertion occurs with probability $p_I$, appending an independent uniform symbol and keeping the same input queued. Deletion occurs with probability $p_D$ and consumes the input. Transmission occurs with probability $p_T=1-p_I-p_D$, consumes the input, and emits it unchanged with probability $1-p_S$, or one of the other three symbols uniformly with probability $p_S$. Processing stops when the last input is consumed. No insertion is made after that time.
\end{definition}
This is precisely the queue convention in~\cite[Sec.~II-A]{Maarouf2023}. Equivalently, each input independently emits $I$ uniform insertions and then either no symbol or a substituted copy, where
\begin{equation}\label{eq:queue}
 \Prob(I=k)=(1-p_I)p_I^k,\qquad
 d:=\Prob(\text{deletion after insertions})=\frac{p_D}{1-p_I},\qquad
 \lambda:=\E|Z|=\frac{1-p_D}{1-p_I}.
\end{equation}
The insertion count and this deletion indicator are independent. $Z$ denotes the emitted word. The observed read is the concatenation of these words; their boundaries are hidden. Notice the distinction between the event probability $p_D$ and the eventual deletion probability $d$.

For input position $i$ and read $j$, let $Z_i^{(j)}$ consist of all symbols emitted while processing input $i$, including preceding insertions and the final transmitted symbol when transmission occurs. Thus $Y^{(j)}=Z_1^{(j)}\cdots Z_n^{(j)}$. The individual words are not observed; only their concatenation is. Under iid input, the tuples $(X_i,Z_i^{(1)},\ldots,Z_i^{(t)})$ are independent across $i$.

For $t$ reads, write $Y^{[t]}=(Y^{(1)},\ldots,Y^{(t)})$. Conditional on $X^n$, the reads use independent channel randomness. Define
\begin{equation}\label{eq:information}
 F_{n,t}=H(X^n\mid Y^{[t]}),\qquad
 J_{n,t}=2-F_{n,t}/n,\qquad J_t=\lim_{n\to\infty}J_{n,t},
\end{equation}
for iid uniform $X_i$. All entropies and logarithms below are in bits. Existence of the limit is proved in \cref{lem:entropy}. Thus $J_{n,t}=I(X^n;Y^{[t]})/n$ and $0\le J_t\le2$, in bits per stored input nucleotide. The normalization divides by the input length, without an additional division by the number of reads or the output length. The term SIR refers to this specified iid input; no assertion that this input maximizes the IDS capacity is needed.

\begin{definition}[Channel likelihood and uniform-input posterior]\label{def:posterior}
Write $\mathcal A^*=\bigcup_{j\ge0}\mathcal A^j$ for all finite output words, including the empty word. For $x^m\in\mathcal A^m$ and $z\in\mathcal A^*$, let $W_m(z\mid x^m)$ be the probability that the channel in \cref{def:channel} emits exactly $z$. For a tuple $z^{[t]}\in(\mathcal A^*)^t$, set
\begin{equation}\label{eq:joint-likelihood}
 W_m^{[t]}(z^{[t]}\mid x^m):=\prod_{j=1}^t W_m(z^{(j)}\mid x^m),\qquad
 D_m(z^{[t]}):=\sum_{a^m\in\mathcal A^m}W_m^{[t]}(z^{[t]}\mid a^m).
\end{equation}
Under iid uniform input, the probability of the output tuple is $4^{-m}D_m(z^{[t]})$. Define
\begin{equation}\label{eq:block-posterior}
 P_m(x^m\mid z^{[t]}):=
 \begin{cases}
 W_m^{[t]}(z^{[t]}\mid x^m)/D_m(z^{[t]}),&D_m(z^{[t]})>0,\\
 0,&D_m(z^{[t]})=0.
 \end{cases}
\end{equation}
For positive $D_m$, this is the posterior under the uniform prior on all $4^m$ input words. For zero $D_m$, the all-zero assignment is a convention for rejecting an impossible output tuple, not a conditional probability distribution. We use $-\log_2 0=+\infty$ when defining posterior thresholds. The dependence of $P_m,D_m$ on the fixed $t$ and channel parameters is suppressed.
\end{definition}

An empty output and an impossible output are different notions: the empty word can have positive probability. Possibility is tested jointly over the $t$ reads. For example, in the noiseless channel, two different output words are individually possible but cannot be two reads of the same input. The posterior in \cref{eq:block-posterior} is defined before imposing code constraints; it is distinct from the posterior over extended codewords used by MAP.

The underlying graph law is the independent socket ensemble of~\cite[Def.~1]{KasaiMAP2026}. The following definition adds the one-bit constraint and known offset.

\begin{definition}[MN socket ensemble and its restriction]\label{def:ensemble}
Let $N=2n$, $K=RN\in\mathbb Z$, and $R=3/\ell$. Uniformly match $\ell K=3N$ punctured sockets to three sockets at each of $N$ checks. Independently match three sockets at each transmitted variable to the other three check sockets. Reduce edge multiplicities modulo two to obtain $A\in\F^{N\times K}$ and $B\in\F^{N\times N}$. Let $G$ denote the two matchings. Set
\begin{align}
 \C_{\rm ext}(G)&=\{(u,v)\in\F^K\times\F^N:Au+Bv=0\},&
 \C(G)&=\{v\in\F^N:\exists u\in\F^K,\ Au+Bv=0\},\label{eq:code}\\
 \C_{\rm ext}^{0}(G)&=\{(u,v)\in\C_{\rm ext}(G):v_1=0\},&
 \C^{0}(G)&=\{v\in\C(G):v_1=0\}.\label{eq:shortened}
\end{align}
The pair $(U,V)$ is uniform on $\C_{\rm ext}^{0}(G)$. Independently choose $C$ uniformly in $\F^N$ and reveal it to encoder and decoder. The stored binary vector is $X=V+C$, and the nucleotide word sent through the channel is $X^n=\beta_n^{-1}(X)$. The coordinate $v_1$ is retained, so the physical length is $n$ nucleotides. For fixed $G,C$, block-MAP uses the matched joint likelihood of all reads.
\end{definition}

The offset uses the random-coset construction~\cite[Def.~3]{BennatanBurshtein2004}. We prove the required independence in \cref{prop:list-map}; the memoryless-channel theorem of that reference is not invoked for the IDS channel.

The projection distinguishes what is stored from the auxiliary variables used to describe it. If $u_0$ is one solution for a fixed $v$, every other solution is $u_0+z$ with $z\in\ker A$. Thus different extended words can produce exactly the same stored sequence. The physical channel cannot distinguish those punctured partners, even from noiseless reads. This is why the actual transmitted rate counts distinct $v$'s, and why recovery of the full extended word additionally requires injectivity of $A$.

All vector sums and matrix operations are over $\F$. The notation $\wt(v)$ counts the nonzero coordinates of $v$, and $\ones_d$ denotes the all-one vector of length $d$.

Every fiber of the projection to $v$ has size $2^{\dim\ker A}$, so $V$ is uniform on $\C^0(G)$, including when either matrix is rank deficient. Odd check degrees on both sides imply $(\ones_K,\ones_N)\in\C_{\rm ext}(G)$. Consequently
\begin{equation}\label{eq:rate}
 k_0(G):=\dim\C^0(G)=\dim\C(G)-1,\qquad
 r_n(G)=k_0(G)/n\longrightarrow 2R.
\end{equation}
\Cref{prop:rate} proves this convergence in $L^1$ directly, with mean absolute deviation $O(\log n/n)$. Conditioning on matrix invertibility would define a different ensemble; it is unnecessary here.

For any fixed $G=g$ and allowed $V=v$, the identity $\Prob(X=x\mid G=g,V=v)=2^{-N}$ follows because exactly one offset, $C=x+v$, produces $x$. Averaging over $V$ therefore leaves the same uniform law for every graph. This uniformity is obtained by averaging over $C$. Once both $G$ and $C$ are fixed, $X$ is instead uniform on the translated code $C+\C^0(G)$. The proof uses the former law to construct graph-independent channel lists, and the latter law when defining the error of each individual code.

\begin{definition}[Full-word MAP error]\label{def:error}
For a fixed graph $g$, offset $c$, and read tuple $y^{[t]}$, the MAP decoder chooses
\begin{equation}\label{eq:map-decoder}
 (\widehat U,\widehat V)_{g,c}(y^{[t]})
 \in\mathop{\rm arg\,max}_{(u,v)\in\C_{\rm ext}^{0}(g)}
 W_n^{[t]}\bigl(y^{[t]}\mid\beta_n^{-1}(v+c)\bigr).
\end{equation}
Use a fixed ordering to resolve ties. For observations with positive probability under this code, likelihood maximization equals posterior maximization because the prior is uniform on the extended code. The same fixed rule defines the decoder on impossible observations. Define
\begin{equation}\label{eq:error-definition}
 p_n(g,c):=\Prob\{(\widehat U,\widehat V)_{g,c}(Y^{[t]})\ne(U,V)\mid G=g,C=c\},
 \qquad \overline P_{\rm block}^{\rm MAP}:=\E_{G,C}p_n(G,C).
\end{equation}
The probability in $p_n(g,c)$ averages over the uniform extended word and channel randomness. A block error means that at least one punctured or transmitted variable is wrong.
\end{definition}

For a general channel $W_N(y\mid x)$ on binary inputs, the same definitions replace the likelihood in \cref{eq:map-decoder} by $W_N(y\mid v+c)$; we write the average error as $\overline P_{\rm block}^{\rm MAP}(W_N)$. This convention is used in \cref{prop:list-map}. The offset is a code-design parameter: the existence proof later fixes one $c$ together with one graph at each length.

\section{Small lists from unaligned DNA reads}\label{sec:lists}
We first use the channel output to narrow the possible stored sequences. The list must contain the true sequence with high probability, yet be small enough that the MN parity constraints can eliminate its other members. The conditional uncertainty per nucleotide tends to $2-J_t$, so the target logarithmic list size is slightly larger than $n(2-J_t)$. The graph estimate in \cref{lem:list} permits lists of size at most $2^{N\eta}$ for any fixed $\eta<1-R$. These requirements are compatible precisely when $J_t>2R$.

The obstacle is that an input block does not have a known output block. We resolve it by trying every consecutive partition of the reads. First we bound the information hidden in those boundaries; then we define the list and prove its two required properties separately.

\subsection{How much information is hidden in the boundaries?}
The reveal-and-hide comparison for output-block boundaries is used in deletion-channel analysis~\cite{FertonaniDuman2010}\cite[Lemma~1]{Dalai2011}. We adapt it to iid uniform input and $t$ IDS reads, including the cost of unbounded output lengths.

\begin{lemma}[Entropy limit and a finite-block bracket]\label{lem:entropy}
For the IDS model with fixed $t$, the limit in \cref{eq:information} exists. With
\begin{equation}\label{eq:geom-entropy}
 g(z)=(z+1)\log_2(z+1)-z\log_2z,\quad g(0)=0,
\end{equation}
every positive integer $m$ satisfies
\begin{equation}\label{eq:bracket}
       J_{m,t}-\frac{t g(m\lambda)}{m}\le J_t\le J_{m,t}.
\end{equation}
Also $J_t$ is nondecreasing in $t$.
\end{lemma}
\begin{proof}
Cut the input after $a$ symbols and reveal, in each read, the output length generated by those symbols. The receiver can then separate two independent iid-input channel blocks. Because this extra information can only reduce conditional entropy, $F_{a+b,t}\ge F_{a,t}+F_{b,t}$. The bounds $0\le F_{n,t}\le2n$ and Fekete's lemma for superadditive sequences give $\lim_n F_{n,t}/n=\sup_m F_{m,t}/m$. In particular, $F_{m,t}/m\le2-J_t$.

To bound the cost of hiding the cuts, partition $qm$ input symbols into $q$ blocks of length $m$. Let $L_{b,j}$ be the number of output symbols generated by block $b$ in read $j$, and let $\boldsymbol L=(L_{b,j})_{b,j}$. The concatenated reads together with $\boldsymbol L$ determine all block outputs, and conversely those outputs determine the concatenated reads and $\boldsymbol L$. Thus $H(X^{qm}\mid Y^{[t]},\boldsymbol L)=qF_{m,t}$. Hiding the lengths increases this entropy by $I(X^{qm};\boldsymbol L\mid Y^{[t]})\le H(\boldsymbol L)$. Subadditivity of entropy therefore gives
\begin{equation}\label{eq:cut-entropy}
 qF_{m,t}\le F_{qm,t}\le qF_{m,t}+qtH(L_{1,1})
                           \le qF_{m,t}+qt g(m\lambda).
\end{equation}
The last bound uses only that $L_{1,1}$ is a nonnegative integer with mean $m\lambda$. We use the relative-entropy method for maximum-entropy bounds~\cite[Sec.~12.1]{CoverThomas2006} and verify the needed discrete case directly. Indeed, for such a variable with mean $\mu>0$, comparison with the geometric law $q_k=(1+\mu)^{-1}(\mu/(1+\mu))^k$ gives $H(L)\le-\E\log_2 q_L=g(\mu)$ by nonnegativity of relative entropy. At $\mu=0$, the variable is zero almost surely. Divide \cref{eq:cut-entropy} by $qm$ and let $q\to\infty$ to obtain \cref{eq:bracket}. Additional reads reduce conditional entropy, proving monotonicity.
\end{proof}

The condition of \cref{thm:main} can be verified whenever a rigorous lower bound $L\le J_t$ exceeds the code rate $6/\ell$. The theorem itself applies throughout $6/\ell<J_t$, independently of the sharpness of a particular bound $L$. The finite-block bracket in \cref{eq:bracket} converts the exact information quantity $J_{m,t}$ into SIR bounds; a numerical evaluation must also account for its computational error. \Cref{app:noise} gives an application using an analytic lower bound.

\subsection{Defining the candidate set from the observed reads}
A candidate list here is a finite set: order and repeated occurrences have no significance. Its elements need not satisfy the code constraints, which will be imposed later. The adjective small refers to its exponential size bound, not to a polynomial-time construction.

For a tentative partition, we keep the input sequences whose blockwise posterior probabilities have a sufficiently large product. This uses the length-$m$ posterior $P_m$ in \cref{def:posterior}, computed under iid uniform input, before imposing any code constraints. We then take the union over all partitions. A word is included if at least one partition admits it.

\begin{definition}[Output partitions and candidate lists]\label{def:channel-list}
Fix an integer $m\ge1$, a number $\epsilon>0$, a length cutoff $\Lambda>\lambda$, and an integer $n_0\ge m$. These are fixed parameters of the list construction, independent of the graph, offset, and observed reads.

For $n\ge n_0$, write $n=qm+s$, where $q=\lfloor n/m\rfloor$ and $0\le s<m$, and put $k=q+\mathbf1_{\{s>0\}}$. For $y^{[t]}\in(\mathcal A^*)^t$, define $\Pi_k(y^{[t]})$ as the set of arrays of cut indices
\begin{equation}\label{eq:partitions}
 \pi=(c_{b,j}):\quad
 0=c_{0,j}\le c_{1,j}\le\cdots\le c_{k,j}=|y^{(j)}|,
 \qquad 1\le j\le t.
\end{equation}
For such a partition, $y_b^{(j)}$ is the substring of $y^{(j)}$ in positions $c_{b-1,j}+1,\ldots,c_{b,j}$, and $y_b^{[t]}=(y_b^{(1)},\ldots,y_b^{(t)})$. Equal cut indices give an empty substring.

The set admitted by $\pi$ is
\begin{equation}\label{eq:partition-list}
 \mathcal L_{n,\pi}(y^{[t]})=
 \left\{x_1\cdots x_q a:\;
 \begin{array}{l}
 x_b\in\mathcal A^m\ (1\le b\le q),\quad a\in\mathcal A^s,\\[2pt]
 \displaystyle\prod_{b=1}^q P_m(x_b\mid y_b^{[t]})
       \ge 2^{-q(F_{m,t}+m\epsilon)}
 \end{array}\right\}.
\end{equation}
Here $\mathcal A^0$ consists of the empty word, and juxtaposition denotes concatenation. If $s>0$, the final output tuple $y_{q+1}^{[t]}$ is not used in the threshold, and every trailing input $a\in\mathcal A^s$ is allowed. With the convention for outputs of marginal probability zero in \cref{def:posterior}, a partition having an impossible full-block output contributes the empty set.

Finally define, for every positive integer $n$ and every observed tuple,
\begin{equation}\label{eq:T-definition}
 \T_n(y^{[t]})=
 \begin{cases}
 \displaystyle\bigcup_{\pi\in\Pi_k(y^{[t]})}\mathcal L_{n,\pi}(y^{[t]}),
   & n\ge n_0\text{ and }\max_j|y^{(j)}|\le\Lambda n,\\
 \varnothing,&\text{otherwise}.
 \end{cases}
\end{equation}
Thus $\T_n(y^{[t]})\subseteq\mathcal A^n$; its binary image is $\beta_n(\T_n(y^{[t]}))\subseteq\F^{2n}$.
\end{definition}

The cuts in \cref{eq:partitions} specify only boundaries between input blocks, not individual insertion or deletion positions. For two blocks and the read $\mathrm{ACG}$, the four partitions are $(\text{empty},\mathrm{ACG})$, $(\mathrm A,\mathrm{CG})$, $(\mathrm{AC},\mathrm G)$, and $(\mathrm{ACG},\text{empty})$. Empty strings must be allowed because a whole input block can disappear. Trying all partitions includes the actual boundaries without giving them to the decoder.

The product in \cref{eq:partition-list} is a scoring distribution for a proposed split. It need not equal the posterior for the unsegmented read. When all full-block outputs are possible, it is nonetheless normalized over $(x_1,\ldots,x_q)$; that fact will bound the number of candidates. At the actual split, the independent channel blocks additionally make this score concentrate. These are distinct reasons for using the same product.

The threshold uses the probability-mass counting principle behind typical-set cardinality bounds~\cite[Thm.~3.1.2]{CoverThomas2006}: if every candidate has probability at least a threshold and total mass is at most one, there are at most the reciprocal threshold many candidates. Here it is applied to a conditional product score for each partition. For the countable output alphabet, we also prove the finite-variance property and the $O(1/n)$ miss bound directly at the true partition.

\begin{lemma}[Channel lists with vanishing miss probability]\label{lem:channel-list}
For every $\eta>1-J_t/2$, the parameters in \cref{def:channel-list} can be chosen so that, under iid uniform input,
\begin{equation}\label{eq:channel-list}
 |\T_n(y^{[t]})|\le2^{2n\eta}\quad\text{for every }y^{[t]},\qquad
 \Prob\{X^n\notin\T_n(Y^{[t]})\}=O(1/n).
\end{equation}
The lists use no graph, offset, or revealed alignment.
\end{lemma}
\begin{proof}
We first choose the parameters, then bound the list size for every observation, and finally bound the probability of losing the true word. Put $\Delta=2\eta-(2-J_t)>0$ and fix any $\Lambda>\lambda$. Choose a finite $m$ large enough that $a_m<\Delta/2$, where
\begin{equation}\label{eq:cut-count}
 a_m:=\frac{t}{m}\log_2\!\bigl(e(1+2\Lambda m)\bigr).
\end{equation}
This is possible because $a_m\to0$. Choose $\epsilon=\Delta/4$ and set
\begin{equation}\label{eq:list-parameters}
 n_0:=\max\{m,\lceil 8(m-1)/\Delta\rceil\}.
\end{equation}
All these quantities are fixed before $n$ grows. The remaining exponent margin exceeds $\Delta/4$ per nucleotide and will absorb the final partial block.

Consider $n\ge n_0$ and an output passing the length cutoff. For one partition, each admitted full-block sequence has product probability at least $2^{-q(F_{m,t}+m\epsilon)}$. Their probabilities sum to at most one, so there are at most $2^{q(F_{m,t}+m\epsilon)}$ such sequences. Allowing every trailing word multiplies this count by $4^s$.

A word of length $j$ has exactly $\binom{j+k-1}{k-1}$ partitions into $k$ consecutive, possibly empty strings: its $k$ substring lengths are nonnegative integers summing to $j$. For both $s=0$ and $s>0$, this count is at most $\binom{j+q}{q}$. Using $\binom{j+q}{q}\le[e(1+j/q)]^q$ and $j\le\Lambda n$, the union in \cref{eq:T-definition} satisfies
\begin{equation}\label{eq:remainder-list}
 \begin{aligned}
 |\T_n(y^{[t]})|
 &\le 2^{q(F_{m,t}+m\epsilon)+2s}
           \bigl[e(1+\Lambda n/q)\bigr]^{tq}\\
 &\le 2^{n(F_{m,t}/m+\epsilon+a_m)+2(m-1)}.
 \end{aligned}
\end{equation}
For the second inequality, $q=\lfloor n/m\rfloor\ge n/(2m)$ gives $n/q\le2m$, while $q\le n/m$ and $s\le m-1$. Thus trying all partitions costs at most $na_m$ bits in the logarithmic list size, and allowing the trailing input costs at most $2(m-1)$ bits. Since $F_{m,t}/m\le2-J_t$, $\epsilon+a_m<3\Delta/4$, and $2(m-1)\le n\Delta/4$ for $n\ge n_0$, the exponent is at most $n(2-J_t+\Delta)=2n\eta$. For outputs failing the cutoff and for $n<n_0$, the empty-list definition already gives the same size bound.

We next check inclusion of the true word. Let $Y_b^{[t]}$ be the outputs actually generated by the $b$th full input block $X_b$. The resulting partition is one of those in \cref{eq:partitions}. Under iid uniform input, the block pairs $(X_b,Y_b^{[t]})$ are iid. Their posterior surprises $Z_b=-\log_2P_m(X_b\mid Y_b^{[t]})$ have mean $F_{m,t}$. A realized word has positive posterior almost surely, so these variables are well-defined. The true word passes the product threshold exactly when
\begin{equation}\label{eq:posterior-list}
 \sum_{b=1}^q Z_b\le q(F_{m,t}+m\epsilon).
\end{equation}
Thus a large posterior surprise means that the block is unusually difficult to identify from its reads; the threshold allows a positive margin above the mean total uncertainty.

Although output lengths are unbounded, $Z_b$ has finite variance. For every possible block-output tuple $z^{[t]}$ and $u\ge0$,
\begin{equation}\label{eq:surprise-tail}
 \Prob\{Z_b>u\mid Y_b^{[t]}=z^{[t]}\}\le\min\{1,4^m2^{-u}\}.
\end{equation}
There are at most $4^m$ input words with posterior less than $2^{-u}$, and their probabilities sum to at most $4^m2^{-u}$. The resulting exponential tail gives a finite second moment, uniformly in the output. Let $\sigma_m^2=\Var(Z_1)<\infty$. Chebyshev's inequality bounds failure of \cref{eq:posterior-list} by $\sigma_m^2/(qm^2\epsilon^2)$.

Let $v_L=\Var(|Z|)<\infty$ for the per-symbol emitted word $Z$ in \cref{eq:queue}. Each read has mean length $n\lambda$ and variance $nv_L$, so the probability that any read exceeds $\Lambda n$ is at most $t v_L/[n(\Lambda-\lambda)^2]$. On the complement of these two failures, the true full-block sequence satisfies the product threshold, and the actual trailing word is allowed. Hence the union contains $X^n$. Therefore, for $n\ge n_0$,
\begin{equation}\label{eq:list-miss-bound}
 \Prob\{X^n\notin\T_n(Y^{[t]})\}
 \le \frac{t v_L}{n(\Lambda-\lambda)^2}
       +\frac{\sigma_m^2}{qm^2\epsilon^2}=O(1/n).
\end{equation}
Here $q=\lfloor n/m\rfloor\ge n/(2m)$ because $n\ge m$. The finitely many lengths below $n_0$ are covered by increasing the constant in the $O(1/n)$ bound. The cutoff only declares failure of this comparison list on rare outputs; it never truncates the physical channel or the MAP likelihood. Incorrect trial partitions can add candidates, but cannot remove the true word once the actual partition admits it.
\end{proof}

\section{A strict candidate-counting bound for MN codes}\label{sec:counting}
We now ask whether the channel list can contain a second MN codeword. A distance estimate and the one-bit restriction first confine every competing difference to an intermediate Hamming-weight interval. We then count parity-compatible differences in an arbitrary list of the allowed size. A uniform negative exponent makes their expected number small; the final comparison decoder turns that count into a block-error bound.

\subsection{Two established ensemble estimates}
Write $h_2(a)=-a\log_2a-(1-a)\log_2(1-a)$, with $0\log_2 0=0$, for binary entropy. Define the check polynomial and its entropy bound by
\begin{align}
 f(x,y)&=\frac{(1+x)^3(1+y)^3+(1-x)^3(1-y)^3}{2},\label{eq:polynomial}\\
 \phi(a,b)&=\inf_{x,y>0}\{\log_2 f(x,y)-3a\log_2x-3b\log_2y\}.\label{eq:phi}
\end{align}
For feasible means, $\phi$ is the maximum entropy of an even-parity six-bit vector whose first and last triples have mean fractions $a,b$. Explicitly, a feasible law is supported on $\{s\in\{0,1\}^6:\sum_{i=1}^6s_i\text{ is even}\}$ and satisfies $\E\sum_{i=1}^3s_i=3a$ and $\E\sum_{i=4}^6s_i=3b$. Symmetrization within each triple makes its individual marginals equal to those means. For infeasible means set $\phi=-\infty$. This identification is the entropy--log-partition duality~\cite[Sec.~3.6, Thm.~3.4(a)]{WainwrightJordan2008}. Boundary means use its limiting extension; a finite positive minimizing pair $x,y$ need not exist there.

The variables $x,y$ in $f$ record the numbers of active sockets in the two triples; they are formal counting variables, not channel symbols. Averaging the two products in \cref{eq:polynomial} cancels exactly the monomials of odd total degree. The coefficient of a surviving monomial $x^i y^j$ counts assignments with $i$ active punctured sockets, $j$ active transmitted sockets, and an even total. Multiplying $N$ copies counts assignments at all checks, while coefficient extraction fixes the total numbers of active sockets. The function $\phi$ bounds the exponential growth of this count at prescribed mean fractions.

\begin{lemma}[Inputs from the MN ensemble analysis]\label{lem:inputs}
For all $a,b\in[0,1]$,
\begin{equation}\label{eq:parity-bound}
 \phi(a,b)\le 2h_2(a)+2h_2(b)+h_2(a)h_2(b).
\end{equation}
For each fixed $\ell\ge4$, there is $0<\delta<1/2$ such that the good-graph event
\begin{equation}\label{eq:good-graph}
 \mathcal G_N:=\{\ker A=\{0\}\text{ and }\wt(v)>\delta N
                   \text{ for every }v\in\C(G)\setminus\{0\}\}
\end{equation}
satisfies $\Prob(\mathcal G_N^c)=O(1/N)$.
\end{lemma}
\begin{proof}
The entropy inequality is~\cite[Lemma~1]{KasaiMAP2026}, with its natural-log normalization divided by $\ln2$. The graph statement is~\cite[Lemma~8]{KasaiMAP2026}. For example, that estimate permits $\delta=1/(20000\ell^3)$ and bounds its exceptional probability by $325\ell^3/N$ plus a polynomial times an exponentially decreasing term. These are channel-independent statements for precisely the independent socket ensemble in \cref{def:ensemble}. The square matrix $B$ need not be invertible.
\end{proof}

\subsection{Why one fixed bit is useful}
On $\mathcal G_N$, every nonzero $v\in\C^0(G)$ obeys
\begin{equation}\label{eq:two-sided-distance}
             \delta N<\wt(v)<(1-\delta)N.
\end{equation}
The lower bound follows from \cref{lem:inputs}. For the upper bound, $v+\ones_N\in\C(G)$ is nonzero, since $v_1=0$. Its weight also exceeds $\delta N$. Thus \cref{lem:list} covers every possible competing transmitted word. Hamming distance is used here to restrict candidate differences in the graph count; it is not an adversarial edit-distance guarantee.

One cannot discard the all-one difference by appealing to a Hamming distance bound. For example, a channel that reveals only the unordered pair $\{x,x+\ones\}$ has mutual information per input bit tending to one under iid uniform input, while any uniform linear code containing $\ones$ has block error at least $1/2$. Fixing one bit removes this ambiguity at a cost of one information bit, or $1/n$ bits per nucleotide ($1/N$ per transmitted bit). This example explains the restriction in a theorem intended to use arbitrary output-dependent lists.

\subsection{Strictness away from the two extreme weights}
For a list exponent $\eta<1-R$, put
\begin{equation}\label{eq:exponent}
 \Psi_\eta(a,b)=(R-3)h_2(a)-3h_2(b)+\phi(a,b)
                       +\min\{h_2(b),\eta\}.
\end{equation}

Here $a$ is the relative weight of a punctured candidate and $b$ is the relative weight of a transmitted difference. Each term in \cref{eq:exponent} contributes to the exponential part $2^{N\Psi_\eta(a,b)}$ of the expected-count bound. The term $Rh_2(a)$ comes from bounding the number of punctured candidates, $\phi(a,b)$ from bounding parity-compatible socket assignments, and $-3h_2(a)-3h_2(b)$ from the two matching denominators. The remaining term counts transmitted candidates, using whichever is smaller: the entire Hamming sphere of relative weight $b$, or the prescribed list. A negative total exponent means that the parity constraints outweigh all these choices.

\begin{lemma}[Uniform negative exponent]\label{lem:gap}
Fix $0\le\eta<1-R$ and $0<\delta<1/2$. There is $\kappa>0$ such that
\begin{equation}\label{eq:negative}
 \sup_{\substack{0\le a\le1\\\delta\le b\le1-\delta}}
                      \Psi_\eta(a,b)\le-\kappa.
\end{equation}
\end{lemma}
\begin{proof}
Write $A_0=h_2(a)$, $B_0=h_2(b)$, and $c_0=1-R$. The parity bound gives
\begin{equation}\label{eq:envelope}
 \Psi_\eta(a,b)\le (B_0-c_0)A_0-B_0+\min\{B_0,\eta\}.
\end{equation}
If $B_0>\eta$, the right side is strictly negative: for $B_0\le c_0$ it is at most $\eta-B_0$, and for $B_0\ge c_0$ it is at most $\eta-c_0$, using $A_0\le1$. If $B_0\le\eta$, it is strictly negative whenever $0<a<1$, because $B_0-c_0<0$.

The coarse bound can equal zero when $A_0=0$. An upper bound of zero would not suffice: the final count still contains a polynomial factor in $N$. We therefore return to the parity distribution at these endpoint means and show that its entropy is strictly smaller than the coarse bound allows.

Only $a=0$ or $a=1$ with $B_0\le\eta<1$ remains. Use the maximum-entropy interpretation of $\phi$ in \cref{eq:phi}. For feasible means, the finite-dimensional set of admissible laws is compact, so entropy attains a maximum. Averaging a maximizing law over permutations within each triple preserves the constraints and cannot decrease entropy; we may therefore take every coordinate of the last triple to have mean $b$.

At $a=0$ or $a=1$, the first triple is deterministic, and the last triple has even or odd parity, respectively. Any two bits determine the third, so $\phi(a,b)\le2h_2(b)$. At $a=0$, equality would require the first two bits of the last triple to be independent Bernoulli$(b)$ variables. Their parity would then have mean $2b(1-b)=b$, forcing $b=0$ or $b=1/2$. At $a=1$, odd parity similarly forces $1-2b(1-b)=b$, hence $b=1/2$ or $b=1$. The weight interval excludes 0 and 1, and $B_0<1$ excludes $1/2$. Thus $\phi(a,b)<2h_2(b)$ at the remaining feasible endpoints, which makes $\Psi_\eta(a,b)<0$ there as well. Infeasible pairs already have exponent $-\infty$.

Pointwise negativity must be strengthened to one bound valid for every weight pair, because the number of pairs grows with $N$. In \cref{eq:phi}, fixing $x,y$ gives a continuous affine function of $(a,b)$; their infimum is upper semicontinuous. Adding the continuous entropy terms preserves this property. On the compact rectangle in \cref{eq:negative}, $\Psi_\eta$ therefore attains a maximum. The endpoint argument ensures that this maximum cannot be zero, even if maximizing weight pairs approach a boundary. Its negative gives a single constant $\kappa>0$ valid throughout the rectangle.
\end{proof}

\begin{lemma}[Arbitrary candidate lists]\label{lem:list}
Fix $0\le\eta<1-R$, $0<\delta<1/2$, a deterministic set $\T\subseteq\F^N$ of cardinality at most $2^{N\eta}$, and $x\in\F^N$. Define the nonnegative integer
\begin{equation}\label{eq:competitor-count}
 M_G(x,\T):=
 \sum_{\substack{x'\in\T\\\delta N\le\wt(x+x')\le(1-\delta)N}}
       \sum_{u\in\F^K}\mathbf1\{Au+B(x+x')=0\}.
\end{equation}
It counts pairs consisting of a candidate transmitted word and a punctured difference satisfying the graph constraints. Here $\mathbf1\{E\}$ is one if $E$ holds and zero otherwise. Then
\begin{equation}\label{eq:list-bound}
 \E_G M_G(x,\T)\le Q_N2^{-N\kappa},\qquad
 Q_N=(K+1)(N+1)(3N+1)^2,
\end{equation}
with $\kappa$ from \cref{lem:gap}. The same bound holds conditionally on a random pair $(x,\T)$ jointly independent of $G$; independence of the two components from each other is not required.
\end{lemma}
\begin{proof}
For fixed $u,z$ of weights $k,w$, the probability of $Au+Bz=0$ is exactly
\begin{equation}\label{eq:socket-count}
 \frac{[x^{\ell k}y^{3w}]f(x,y)^N}
             {\binom{3N}{\ell k}\binom{3N}{3w}}.
\end{equation}
The independent matchings make the active row sockets uniform subsets of the prescribed sizes. The numerator counts those satisfying the even-parity condition at every check, including multiplicities that cancel modulo two. This is the two-type configuration count used in~\cite[Sec.~4.2, Eq.~(11), and App.~B, Eq.~(22)]{KasaiMAP2026}.

The coefficient bound used below is deterministic. Since all coefficients of $f^N$ are nonnegative, its selected coefficient times $x^{\ell k}y^{3w}$ is at most $f(x,y)^N$ for every positive $x,y$. Dividing by that monomial and taking the infimum gives the exponent $N\phi(k/K,w/N)$, since $\ell k=3N(k/K)$. Thus the generating function supplies an upper bound for each weight pair before we sum over candidates.

There are $\binom Kk$ choices of $u$, and at most $\min\{\binom Nw,2^{N\eta}\}$ words of the list at distance $w$ from $x$. Using the type-class estimates~\cite[Example~11.1.3, Eq.~(11.40)]{CoverThomas2006}
\begin{equation}\label{eq:binomial}
 (m+1)^{-1}2^{m h_2(j/m)}\le\binom mj\le2^{m h_2(j/m)}
\end{equation}
and the coefficient bound from \cref{eq:phi}, the term at $(k,w)$ is at most $(3N+1)^2 2^{N\Psi_\eta(k/K,w/N)}$. There are at most $(K+1)(N+1)$ terms. Apply \cref{eq:negative}.
\end{proof}

\subsection{From candidate lists to block-MAP reliability}
The graph estimate can be separated from the construction of lists for a particular channel. The channel lists already provide the required size and miss bounds; the following proposition explains how these properties imply MAP reliability.

\begin{proposition}[List-to-MAP bound]\label{prop:list-map}
Use the ensemble and uniform known offset of \cref{def:ensemble}. Let $W_N(y\mid x)$ be any channel from $\F^N$ to a countable output alphabet, with conditional law depending on the graph and offset only through its input $x$. Fix $0\le\eta<1-R$. Suppose deterministic lists $\mathcal L_N(y)\subseteq\F^N$ satisfy $|\mathcal L_N(y)|\le2^{N\eta}$ for every $y$ and have miss probability $\epsilon_N$ under uniform input. With $\delta$ from \cref{lem:inputs}, $\kappa$ from \cref{lem:gap}, and $Q_N$ from \cref{eq:list-bound}, the ensemble-average full extended-word error satisfies
\begin{equation}\label{eq:list-map}
 \overline P_{\rm block}^{\rm MAP}(W_N)
 \le \Prob(\mathcal G_N^c)+\epsilon_N+Q_N2^{-N\kappa}.
\end{equation}
\end{proposition}
\begin{proof}
We exhibit a decoder whose failures can be counted. For fixed $g,c,y$, form the set of extended words consistent with both the code and the list,
\begin{equation}\label{eq:comparison-decoder}
 \mathcal D_{g,c}(y):=
 \{(u,v)\in\C_{\rm ext}^{0}(g):v+c\in\mathcal L_N(y)\}.
\end{equation}
The comparison decoder returns the unique element of this set if its cardinality is one, and declares a block error otherwise. If the true input lies in the list, the true extended word belongs to this set. We need to bound the chance that a second member remains.

First, the uniform offset makes $X$ uniform for every graph. Because the channel depends only on $X$, we have $\Prob(X=x,Y=y\mid G=g)=2^{-N}W_N(y\mid x)$. Hence $(X,Y)$, and therefore $(X,\mathcal L_N(Y))$, is jointly independent of $G$. In the analysis we may condition on $(X,Y)$ and still use the original graph distribution in \cref{lem:list}. We do not also condition on $C$ or on the good-graph event.

Now let $(U,V)$ be the true extended word and $(u',v')$ a competing member of \cref{eq:comparison-decoder}. Write $x'=v'+C$. On a good graph, injectivity of $A$ rules out a second punctured partner with $v'=V$, so $x'\ne X$. The common offset cancels, and the parity equations give
\begin{equation}\label{eq:competitor-difference}
 X+x'=V+v'\in\C^0(G)\setminus\{0\},\qquad
 A(U+u')+B(X+x')=0.
\end{equation}
The first relation puts the difference inside the weight interval in \cref{eq:two-sided-distance}. The second identifies the punctured vector counted in \cref{lem:list}: it is the difference $U+u'$, not the competing component $u'$ itself.

For fixed true $(U,V),C$, different competitors give different counted pairs $(x',U+u')$. Thus their number on a good graph is at most the graph count $M_G(X,\mathcal L_N(Y))$ in \cref{eq:competitor-count}. In particular,
\begin{equation}\label{eq:competitor-probability}
 \Prob\{\mathcal G_N\text{ and a competitor exists}\}
 \le\E M_G(X,\mathcal L_N(Y))\le Q_N2^{-N\kappa}.
\end{equation}
The first inequality uses that a nonnegative integer is at least one whenever a competitor exists. The second applies \cref{lem:list} after conditioning on $(X,Y)$; the count is averaged over all graphs, so no conditioned socket ensemble is introduced.

The comparison decoder succeeds whenever the graph is good, the true input is in the list, and no competitor exists. A union bound over the complementary events gives \cref{eq:list-map}. For every fixed graph and offset, matched block-MAP minimizes the full-word error under the same prior, so its error is no greater than this comparison decoder's error. This proves the stated MAP bound.
\end{proof}

\section{Proof of the main theorem}\label{sec:map-proof}
The channel list identifies a small candidate set of physical inputs from the observed reads. The MN estimate then controls whether that set contains another codeword. These two statements give a decoder with small block error; MAP optimality and a final selection argument give the conclusions of \cref{thm:main}.

\begin{proof}[Proof of \cref{thm:main}]
Assume $J_t>2R$ and choose
$1-J_t/2<\eta<1-R$.
The left inequality permits the construction in \cref{lem:channel-list}, while the right inequality permits the graph count in \cref{lem:list}. Thus the same exponent satisfies both requirements. For each observed tuple $y^{[t]}$, construct the nucleotide candidate set $\T_n(y^{[t]})$ and use its binary image
$\mathcal L_N(y^{[t]})=\beta_n(\T_n(y^{[t]}))$, where $N=2n$.
The bijection $\beta_n$ preserves cardinality. Hence this binary list has at most $2^{N\eta}$ elements, and its miss probability under iid-uniform input is $O(1/n)$.

The IDS channel has a countable output alphabet and depends on the graph and offset only through the stored input. We can therefore apply the list-to-MAP bound in \cref{prop:list-map} to these binary lists. Its three terms account for an exceptional graph, a missed true input, and a surviving competitor:
\begin{equation}\label{eq:main-error}
 \begin{aligned}
 \overline P_{\rm block}^{\rm MAP}
 &=\E_{G,C}p_n(G,C)\\
 &\le \Prob(\mathcal G_N^c)
       +\Prob\{X^n\notin\T_n(Y^{[t]})\}
       +Q_N2^{-N\kappa}
 =O(1/n).
 \end{aligned}
\end{equation}
The first term is $O(1/n)$ by the graph estimate in \cref{lem:inputs}, and the second has that order by \cref{lem:channel-list}. The third decays exponentially up to the polynomial factor $Q_N$. All constants are fixed before $n$ grows.

It remains to select deterministic codes with both the required rate and the error bound. Restrict this final selection to good graphs whose actual rate is close to $2R$:
\begin{equation}\label{eq:rate-selection-event}
 \mathcal H_n:=\mathcal G_N\cap\{|r_n(G)-2R|\le n^{-1/2}\}.
\end{equation}
The exceptional-graph estimate in \cref{lem:inputs} and the actual-rate bound in \cref{prop:rate}, followed by Markov's inequality, give $\Prob(\mathcal H_n^c)\le\Prob(\mathcal G_N^c)+\sqrt n\,\E|r_n(G)-2R|=O(\log n/\sqrt n)\to0$. Thus $\Prob(\mathcal H_n)\ge1/2$ for every sufficiently large admissible length. Since the error $p_n(G,C)$ is nonnegative,
\begin{equation}\label{eq:selection-average-error}
 \E_{G,C}[p_n(G,C)\mid\mathcal H_n]
 \le \frac{\E_{G,C}p_n(G,C)}{\Prob(\mathcal H_n)}
 \le 2\overline P_{\rm block}^{\rm MAP}
 =O(1/n).
\end{equation}
At least one graph--offset pair $(g_n,c_n)$ in this event has error no larger than the conditional average. Selecting such a pair at each sufficiently large admissible length gives $p_n(g_n,c_n)=O(1/n)$ and $|r_n(g_n)-2R|\le n^{-1/2}\to0$ simultaneously. Since $g_n\in\mathcal G_N$, its matrix $A$ is injective and each transmitted word has a unique punctured partner. Conditioning on $\mathcal H_n$ is used only for this final selection; the graph count establishing \cref{eq:main-error} uses the original unconditioned ensemble. The selected graph and offset are then fixed and may depend on the channel and the number of reads. Their error guarantee averages over uniform extended codewords and channel outputs, as in \cref{def:error}, rather than taking the maximum over messages. Together with \cref{prop:rate}, this proves all claims of \cref{thm:main}.
\end{proof}

The candidate-list argument controls block error directly. A normalized conditional-entropy bound alone would not establish reliability: uncertainty between the all-zero and all-one words has entropy one bit, yet every bit has error probability $1/2$ under the uniform prior on those two words.

\section{Conclusion}
For each fixed $\ell\ge4$, the $(\ell,3,3)$ MN socket ensemble with one transmitted bit fixed and a known uniform offset has vanishing block-MAP error on the quaternary queue IDS channel whenever its limiting rate $6/\ell$ is strictly below the joint iid-uniform SIR. The proof controls competing codewords directly, with hidden output partitions contributing a vanishing rate overhead. This extends the memoryless-channel MAP analysis of MN codes to unaligned reads with unbounded insertion runs and drift.

The MAP guarantee is asymptotic, with the degrees and channel parameters held fixed as the blocklength grows. The DNA model describes an identified strand with a fixed number of independent reads. Exact balance of the combined guanine and cytosine content, homopolymer constraints, random sampling depth, strand dropout, and context-dependent sequencing require either a replacement small-list theorem with the appropriate input law or an additional outer storage architecture. Sharper finite-length bounds and the case where the rate equals the SIR also remain to be determined.

\Cref{prop:list-map} separates the ensemble analysis from the IDS specialization. Applying it to another channel requires sufficiently small uniform-input lists with vanishing miss probability, constructed from the observations available to the receiver.

\appendix
\crefalias{section}{appendix}
\section{Convergence of the actual rate}\label{app:rate}
Counting variables minus checks gives the design dimension $K$, but need not give the transmitted dimension of an individual graph. Dependent checks can increase the dimension of $\C_{\rm ext}(G)$, whereas a nontrivial $\ker A$ identifies several extended words under projection. The next proof bounds both corrections in expectation. This is sufficient for convergence of the actual rate, even without requiring exact ranks for every graph.

\begin{proposition}[Actual-rate convergence]\label{prop:rate}
With $Q_N$ from \cref{eq:list-bound}, let $S_N=(K+1)(3N+1)$. The actual rate satisfies
\begin{equation}\label{eq:rate-l1}
 \E|r_n(G)-2R|\le
 \frac{\log_2 Q_N+\log_2 S_N+1}{n}
 =O(\log n/n).
\end{equation}
\end{proposition}
\begin{proof}
Sum \cref{eq:socket-count} over all pairs $(u,v)$ and use \cref{eq:binomial,eq:parity-bound}. For $a=k/K$, $b=w/N$, the exponential part of the expected weight count is at most
\begin{equation}\label{eq:rate-spectrum}
 (R-3)h_2(a)-2h_2(b)+\phi(a,b)
 \le h_2(a)\bigl(h_2(b)-(1-R)\bigr)\le R.
\end{equation}
Thus $\E|\C_{\rm ext}(G)|\le Q_N2^K$. Restricting the count to $v=0$ leaves only one binomial denominator and gives $\E|\ker A|\le S_N$, since its exponent is at most $(R-1)h_2(a)\le0$.

Write $d_{\rm ext}=\dim\C_{\rm ext}(G)$ and $q=\dim\ker A$. There are $N$ checks on $K+N$ variables, so $d_{\rm ext}\ge K$, while the projection has dimension $d_{\rm ext}-q$. Jensen's inequality gives $\E(d_{\rm ext}-K)\le\log_2 Q_N$ and $\E q\le\log_2 S_N$. Since $k_0=d_{\rm ext}-q-1$, we have $|k_0-K|\le(d_{\rm ext}-K)+q+1$. Dividing by $n$ proves \cref{eq:rate-l1}, and Markov's inequality gives convergence in probability. This proof does not require a full-rank event.
\end{proof}

\section{An application using an analytic lower bound}\label{app:noise}
\Cref{thm:main} uses the true SIR $J_t$. The following lower bound provides one way to verify its condition for a particular channel. Here $d=p_D/(1-p_I)$ is the eventual deletion probability defined in \cref{eq:queue}. With no insertions or deletions and a single read, the quaternary symmetric channel has the exact SIR $J_1=2-h_2(p_S)-p_S\log_2 3$.

For a single read, revealing insertion counts and deletion indicators aligns all surviving transmissions. Those indicators are independent of the input. The conditional mutual information is $n(1-d)[2-h_2(p_S)-p_S\log_2 3]$, and the revelation costs at most $n[h_2(p_I)/(1-p_I)+h_2(d)]$. 

Writing the revealed data as $S$, independence of $S$ from $X^n$ gives $I(X^n;Y,S)=I(X^n;Y\mid S)$. Hiding $S$ loses $I(X^n;S\mid Y)\le H(S)$ bits. Subtracting this entropy from the aligned-channel information therefore gives a lower bound on the information in the original read, which explains the direction of the following inequality:

\begin{equation}\label{eq:explicit-lower}
 J_1\ge J_{\rm low}:=
 \max\left\{0,(1-d)[2-h_2(p_S)-p_S\log_2 3]
                -\frac{h_2(p_I)}{1-p_I}-h_2(d)\right\}.
\end{equation}
Additional reads cannot decrease mutual information, so $J_t\ge J_1\ge J_{\rm low}$. Thus $J_{\rm low}>6/\ell$ verifies the sufficient condition of \cref{thm:main} for every fixed number of reads, without evaluating $J_t$ exactly. This bound uses an alignment genie only to lower-bound the information rate, by subtracting its entropy. It provides no alignment information to the actual decoder.

For example, $p_I=p_D=p_S=0.01$ gives $d=1/99$. Substituting $h_2(0.01)<0.081$, $h_2(1/99)<0.082$, and $\log_2 3<1.585$ into \cref{eq:explicit-lower} gives $J_{\rm low}>1.72$. The $(6,3,3)$ family, whose limiting rate is $6/6=1$ bit per nucleotide, therefore satisfies \cref{thm:main} for every fixed $t\ge1$. This verification requires no numerical SIR estimate.

\end{document}